\documentclass[a4paper,12pt]{article}
\usepackage[utf8]{inputenc}
\usepackage{a4wide}
\usepackage[T1]{fontenc}
\usepackage{lmodern}
\usepackage{amsmath,amssymb}
\usepackage[a4paper]{geometry}
\usepackage{amsthm}
\usepackage[english]{babel}
\usepackage{graphicx}
\usepackage{dsfont}
\usepackage{url}
\usepackage{relsize}
\usepackage{caption} 
\everymath{\displaystyle}

\newtheorem{prop}[subsubsection]{Proposition}

\theoremstyle{plain}
\newtheorem{lem}[subsubsection]{Lemma}
\theoremstyle{plain}
\newtheorem{thm}[subsubsection]{Theorem}
\theoremstyle{definition}
\newtheorem{defi}[subsubsection]{Definition}
\theoremstyle{remark}

\newcommand{\E}{\mathbb{E}}
\newcommand{\ke}{\text{Ker}}
\newcommand{\im}{\text{Im}}
\newcommand{\C}{\mathbb{C}}
\newcommand{\R}{\mathbb{R}}
\newcommand{\Z}{\mathbb{Z}}

\newcommand{\Pro}{\mathbb{P}}
\newcommand{\Hb}{\mathcal{H}}

\newcommand{\Bb}{\mathcal{B}}
\newcommand{\M}{\mathcal{M}}

\newcommand{\F}{\mathcal{F}}
\newcommand{\Sb}{\mathcal{S}}

\newcommand{\Sp}{\textrm{Sp}}
\newcommand{\Tr}{\mathrm{Tr}}

\newcommand{\ii}{\mathrm{i}}
\newcommand{\Lin}{\mathcal{L}}
\newcommand{\der}{\mathrm{d}}
\newcommand{\Der}{\mathrm{D}}
\newcommand{\ind}{\mathds{1}}
\newcommand{\A}{\mathcal{A}}

\usepackage{fancyhdr}
\begin{document}

\title{Central Limit Theorem and Large Deviation Principle for Continuous Time Open Quantum Walks}

\author{Hugo BRINGUIER\\ \\ Institut de Mathématiques de Toulouse ; UMR5219,\\ Université de Toulouse ; CNRS,\\UPS IMT, F-31062 Toulouse Cedex 9, France\\}

\maketitle
\begin{abstract}\textit{Open Quantum Walks} (OQWs), originally introduced in \cite{MR2930584}, are quantum generalizations of classical Markov chains.  Recently, natural continuous time models of OQW have been developed in  \cite{pellegrini2014continuous}. These models, called Continuous Time Open Quantum Walks (CTOQWs), appear as natural continuous time limits of discrete time OQWs. In particular they are quantum extensions of continuous time Markov chains. This article is devoted to the study of homogeneous CTOQW on $\Z^d$. We focus namely on their associated quantum trajectories which allow us to prove a Central Limit Theorem for the "position" of the walker as well as a Large Deviation Principle.


\end{abstract}
\section{Introduction}

Open Quantum Walks concern evolution on lattices driven by quantum operations. They describe Markovian dynamics influenced by internal degrees of freedom. They have been introduced originally by \cite{MR2930584} (see also \cite{MR2432029}). These OQWs are promising tools to model physical problems, especially in computer science (see \cite{sinayskiy2012efficiency}). They can also model a variety of phenomena, as energy transfer in biological systems (\cite{marais2013decoherence}).
 
Continuous time models have been developed as a natural continuous time limit of discrete time models \cite{pellegrini2014continuous, MR3268105}. In particular in \cite{MR3268105}, a natural extension of Brownian motion called Open Quantum Brownian Motion has been constructed. In this article, we focus on the continuous time open quantum walks (CTOQWs) model presented in \cite{pellegrini2014continuous}. More precisely, we focus on CTOQWs on $\mathbb Z^d$. Briefly speaking, CTOQWs on $\mathbb Z^d$ concern the evolution of density operators of the form 

\begin{equation}\label{diags}
\mu=\sum_{i\in\Z^d}\rho(i)\otimes|i\rangle\langle i|\in\Hb\otimes\C^{\Z^d}
\end{equation}
where the "$\Z^d$-component" represents the "position" of the walker and $\Hb$ is a Hilbert space describing the internal degrees of freedom. In particular, if $\mathcal D$ denotes the set of density operators of the form \eqref{diags}, CTOQWs are described by a semigroup $\{\phi_t\}$ such that, $\phi_t$ preserves $\mathcal D$ for all $t\geq 0$.

In the context of quantum walks, one is mainly interested in the position of the walker. At time $0$, starting with density matrix in $\mathcal D$ as \eqref{diags}, the quantum measurement of the "position" gives rise to a probability distribution $q_0$ on $\Z^d$, such that, for all $i\in\Z^d$,
\[q_0(i)=\Pro("\text{that the walker is in }i")=\Tr(\rho(i)) \ .\] 
As well, after evolution, if
\[\mu_t=\phi_t(\mu)=\sum_{i\in\Z^d}\rho^{(t)}(i)\otimes|i\rangle\langle i|\]
then
\[q_t(i)=\Pro("\text{that the walker, at time}\, t, \text{ is in }i")=\Tr(\rho^{(t)}(i)) \ .\] 
In \cite{pellegrini2014continuous}, it has been shown that usual classical continuous time Markov chains are particular cases of CTOQWs. In particular one can easily construct models where the distribution $q_t$ corresponds to the one of a classical continuous time Markov chain. Contrary to continuous time Markov chains, the distribution $q_t$ of CTOQWs cannot be in general recovered by the knowledge of the initial distribution $q_0$. One needs to have access to the full knowledge of the initial state $\mu$. In this sense, this justifies the name \textit{quantum walks}. 

Our models of continuous time quantum walks are rather different from the usual models of unitary quantum walks. An essential difference concerns the large time behaviour of the corresponding distribution $q_t$. Let $Q_t$ be a random variable of law $q_t$, in the unitary quantum walk theory it has been shown that $(Q_t)$ satisfies a Central Limit Theorem of the type
$$\frac{Q_t}{t}\underset{t\rightarrow\infty}{\longrightarrow}\tilde{Q} \, ,$$
where $\tilde Q$ has distribution
$$\frac{dx}{\pi\sqrt{1-x^2}}.$$
Note that such behaviour is not usual in classical probability where usually one expects speed in $\sqrt t$ and Gaussian law as limit in the Central Limit Theorem (CLT). In our context, the distributions $(q_t)_{t\geq 0}$ express a rather classical behaviour  in large time in the sense that a more usual CLT holds. In particular this paper is devoted to show that for CTOQWs one has the following weak convergence
$$\frac{Q_t-m}{\sqrt{t}}\underset{t\rightarrow\infty}{\longrightarrow}\mathcal N(0,\sigma^2)\, ,$$
where $\mathcal N(0,\sigma^2)$ denotes usual Gaussian law. Such phenomena have also been observed in the discrete setting of OQWs \cite{MR3296637}. A key point to show this result is the use of the quantum trajectories associated to the CTOQWs. In general, quantum trajectories describe evolutions of quantum system undergoing indirect measurements (see \cite{barchielli2009quantum} for an introduction). In the context of CTOQWs, quantum trajectories describe the evolution of the states undergoing indirect measurements of the position of the walker. In particular these quantum trajectories appear as solution of jump-type stochastic differential equations called stochastic master equations (see \cite{pellegrini2014continuous} for link between discrete and continuous time models in the context of OQW, one can also consult \cite{MR3268105} for such an approach in the context of Open Quantum Brownian Motion). In the physic literature, note that such models appear also naturally in order to describe non-Markovian evolutions. They are called \textit{non-Markov generalization of Lindblad equations} (see \cite{breuer2007non, MR2546050, MR2759470}). 

After establishing the CLT, our next goal is to investigate a Large Deviation Principle (LDP) for the position of the walker. In particular under additional assumptions, one can apply the G\"artner-Ellis Theorem in order to obtain the final result (one can consult \cite{carbone2014homogeneous} for a similar result for discrete time OQWs).
\bigskip

The article is structured as follows. In Section 2, we present the model of CTOQWs on $\mathbb Z^d$. Next we develop the theory of quantum trajectories which describe the continuous measurement of the position. In Section 3, we present the Central Limit Theorem. Section 4 is devoted to the Large Deviation Principle (LDP). Finally in Section 5, we present some examples which illustrate the CLT and the LDP.

\section{Continuous Time Open Quantum Walks}

\subsection{Main setup}

The models of Continuous Time Open Quantum Walks have been formalized in \cite{pellegrini2014continuous}. They arise as continuous limits of discrete time OQWs (we do not recall the discrete time models and we refer to \cite{MR2930584}). These limits processes are described by particular types of \textit{Lindblad master equations}. Originally, these equations appear in the "non-Markovian generalization of Lindblad  theory" from Breuer \cite{breuer2007non}. In this article, we focus on nearest neighbors, homogeneous CTOQWs on $\Z^d$.

In the sequel, $\Hb$ denotes a finite dimensional Hilbert space and $\Sb_\Hb$ denotes the space of density matrix on $\Hb$: \[\mathcal{S_{\Hb}}=\{\rho\in\Bb(\Hb)\mid\rho^*=\rho,\rho\geq 0 , \Tr(\rho)=1 \}.\]
We put $\mathcal K_d=\mathcal H\otimes\mathbb C^{\mathbb Z^d}$ where $\mathbb C^{\mathbb Z^d}$ stands for the position of a particle while $\mathcal H$ corresponds to the internal degree of freedom of this particle. We consider the canonical basis $\{e_1,...,e_d\}$ of $\Z^d$, we set $e_0=0_d$ and $e_{d+r}=-e_r$ for all $r\in\{1,...,d\}$. The canonical basis of $\mathbb C^{\mathbb Z^d}$ is denoted by $({\vert i\rangle})_{i\in\mathbb Z^d}$.\\

As announced we focus on particular diagonal density matrices of $\mathcal K_d$:
$$\mathcal D=\left\{\mu\in\mathcal B(\mathcal K_d), \mu=\sum_{i\in\mathbb Z^d}\rho(i)\otimes\vert i\rangle\langle i\vert, \rho(i)\geq0,\sum_{i\in\mathbb Z^d}\textrm{tr}\big(\rho(i)\big)=1\right\}.$$

In the sequel we shall consider evolutions on $\mathcal K_d$ which preserve $\mathcal D$. To this end we consider a family of operators $\{D_r\}_{r=1,\ldots,2d}$ on $\mathcal B(\mathcal H)$ 
and we define the operators $\{B_i^r\}_{r=1,\ldots,2d}$ on $\mathcal B(\mathcal K_d)$ such that 
$B_i^r=D_r\otimes\vert i+e_r\rangle\langle i\vert$.\\

Now as announced the CTOQWs are generated by particular Lindblad master equations. Let $\M_c$ the following Lindblad operator on $\Hb\otimes\C^{\Z^d}$, 
\[\begin{array}{llll}\M_c:&\Bb(\Hb\otimes\C^{\Z^d})&\rightarrow&\Bb(\Hb\otimes\C^{\Z^d})\\
&\mu&\mapsto &-\ii [H\otimes I,\mu]+ \displaystyle{\sum_{i\in\mathbb Z^d}\sum_{r=1}^{2d}}  \big({B_i^r}  \mu B_i^{r*}-\frac{1}{2}\{{B_i^{r*} B_i^r},  \mu \}\big)

\end{array}\]
where
$H$ is a self-adjoint operator on $\Hb$ which is called the Hamiltonian.\\

Let us introduce the operator
$$D_0=-\ii H-\frac{1}{2}\sum\limits_{r=1}^{2d}  {D_r^* D_r} \, .$$
The next computation shows that $\mathcal M_c$ preserves the set $\mathcal D$.

\begin{align*}
 \M_c(\mu)&= \sum\limits_{i\in\Z^d}\Bigg[-\ii[H,{\rho(i)}]\otimes|i\rangle \langle i| +\sum\limits_{r=1}^{2d}D_r\rho(i)D_r^{*}\otimes|i+e_r\rangle\langle i+e_r|\\
   \hphantom{\M_c(\mu)=}&  \hphantom{\sum\limits_{i\in\Z^d}\Big(-\ii[H,{\rho(i)}]\otimes|i\rangle \langle i| + \, \, \,}-\frac{1}{2}\sum\limits_{r=1}^{2d}\{D_r^*D_r,\rho(i)\}\otimes|i\rangle\langle i|\Bigg]\\
&=\sum\limits_{i\in\Z^d}\Bigg[\big(D_0\rho(i)+\rho(i)D_0^*\big)\otimes|i\rangle\langle i| +\sum\limits_{r=1}^{2d}D_r\rho(i)D_r^{*}\otimes|i+e_r\rangle\langle i+e_r|\Bigg] \\
&=\sum\limits_{i\in\Z^d}\Bigg(D_0\rho(i)+\rho(i)D_0^* +\sum\limits_{r=1}^{2d}D_r\rho(i-e_r)D_r^{*}\Bigg)\otimes|i\rangle\langle i| \, , \end{align*}
for all $\mu=\sum_{i\in\mathbb Z^d}\rho(i)\otimes\vert i\rangle\langle i\vert$.
\\



The following proposition describes precisely our model of CTOQWs.
\begin{prop}{\cite{pellegrini2014continuous}}\label{proppel}
Let $\mu^{(0)}=\displaystyle{\sum_{i\in\Z^d}}\rho^{(0)}(i)\otimes|i\rangle\langle i| $, the equation
 \begin{equation}\label{ctoqw}
 \frac{\der}{\der t}\mu^{(t)}=\M_c(\mu^{(t)}),
 \end{equation}
  with initial condition $\mu^{(0)}$ admits a unique solution $(\mu^{(t)})_{t\geq 0}$ with values in $\mathcal D$.

 More precisely, $\mu^{(t)}$ is of the form $\mu^{(t)}=\sum\limits_{i\in\Z^d}\rho^{(t)}(i)\otimes|i\rangle\langle i| $ such that:
\[\frac{\der}{\der t}\rho^{(t)}(i)=D_0\rho^{(t)}(i)+\rho^{(t)}(i)D_0^* +\sum_{r=1}^{2d}D_r\rho^{(t)}(i-e_r)D_r^{*}  \, ,\]
for all $i\in\mathbb Z^d$.
\end{prop}

\begin{defi}
The evolution \eqref{ctoqw} is called a Continuous Time Open Quantum Walk on $\mathbb Z^d$.
\end{defi}

This definition is justified by the following.  The operator $B_i^r$ transcribes the idea that the particle localized in $|i\rangle $ can only jump to one of its nearest neighbors $|i+e_r\rangle$, and in this case, the transformation on $\Hb$ is governed by $D_r $. In the case the particle stands still, the evolution on $\Hb$ is governed by $D_0$. It is the exact analogue of the usual OQWs for continuous time evolutions. An interesting fact has been pointed out in \cite{pellegrini2014continuous}, usual continuous time classical Markov chains can be realized within this setup.

Now let us describe the probability distributions associated to CTOQWs. 

\begin{defi}Let $\mu^{(0)}=\displaystyle{\sum_{i\in\Z^d}}\rho^{(0)}(i)\otimes|i\rangle\langle i| $. Let $\mu^{(t)}=\displaystyle{\sum_{i\in\Z^d}}\rho^{(t)}(i)\otimes|i\rangle\langle i|$ be the solution of the equation
 \begin{equation*}
 \frac{\der}{\der t}\mu^{(t)}=\M_c\big(\mu^{(t)}\big) \, .
 \end{equation*}
 We define
 \begin{equation}
 q_t(i)=\Tr\big[\mu^{(t)}\left(I\otimes\vert i\rangle\langle i\vert\right)\big]=\Tr\big[\rho^{(t)}(i)\big]
 \end{equation}
and we denote $Q_t$ the random variable on $\Z^d$ of law $q_t$, that is
$$\mathbb P[Q_t=i]=q_t(i),$$
for all $i\in\mathbb Z^d$.
\end{defi}

As we can see in Section 3 and as it was announced in the introduction, the shape of $q_t$ seems to converge to Gaussian shape. This is exactly the result pointed out by the CLT in Section 3. In order to prove this, we shall need the theory of quantum trajectories for CTOQWs. 

\subsection{Quantum trajectories} 
As in the discrete case, quantum trajectories are essential tools for showing the CLT and the LDP. The description of quantum trajectories is less straightforward than the one in OQWs. It makes use of stochastic differential equations driven by jump processes. We refer to \cite{pellegrini2014continuous} for the justification of the below description and the link between discrete and continuous time models. One can also consult \cite{breuer2007non} where general indirect measurements for non-markovian generalization of Lindblad equations have been developped. 

\begin{prop}\label{omega}
Let $\mu^{(0)}=\sum_{i\in\Z^d}\rho^{(0)}(i)\otimes|i\rangle\langle i| $ be an initial state on $\Hb\otimes\C^{\Z^d}$. 
The quantum trajectory describing the indirect measurement of the position of the CTOQWs led by $\mathcal M_c$ is modeled by a Markov process $\left(\omega^{(t)}=\rho_t \otimes \vert X_t\rangle\langle X_t\vert\right)_{t\geq 0}$. This Markov process is valued in the set
$$\mathcal P=\big\{\rho\otimes\vert i\rangle\langle i\vert,\rho\in\mathcal S_{\mathcal H}, i\in\mathbb Z^d\big\} $$ such that \[\omega^{(0)}=\frac{\rho^{(0)}(i)}{\Tr\big(\rho^{(0)}(i)\big)}\otimes|i\rangle\langle i| \text{ with probability } \Tr\big(\rho^{(0)}(i)\big)\] and such that the following differential equation is satisfied:
\begin{align}\label{qtcc}
 \omega^{(t)}=&\,  \omega^{(0)}+\int_{0}^t\Big(D_0\rho_{s-}+\rho_{s-}D_0^*-\rho_{s-}\Tr(D_0\rho_{s-}+\rho_{s-}D_0^*)\Big)\otimes|X_{s-}\rangle\langle X_{s-}| \ \der s \nonumber\\
 &+\sum_{r=1}^{2d}\int_{0}^{t} \int_{\R}   \Bigg(\frac{D_r\rho_{s-}D_r^*}{\Tr(D_r\rho_{s-}D_r^*)}\otimes|X_{s-}+e_r\rangle\langle X_{s-}+e_r|\nonumber\\
 &\hphantom{+\mathlarger{\mathlarger{ \sum}}_{r=1}^{2d} \, lol   \Bigg( \, \, \,lolollol \,}-\rho_{s-}\otimes|X_{s-}\rangle\langle 
X_{s-}|\Bigg)\ind_{0<y<\Tr(D_r^{}\rho_{s-}D_r^*)} N^r(\der y, \der s)
\end{align}
where $\{N^r\}_{r\in\{1,\ldots,2d\}}$ are independent Poisson point processes on $\R^2$.

In particular the Markov process $(\rho_t,X_t)_{t\geq 0}$ is valued in $\mathcal S_{\mathcal H}\times\mathbb Z^d$ and satisfies 

\begin{eqnarray}\label{eqcul}
 \der\rho_{s}&=& \Big(D_0\rho_{s-}+\rho_{s-}D_0^*-\rho_{s-}\Tr(D_0\rho_{s-}+\rho_{s-}D_0^*)\Big)\der s \nonumber\\
 && + \sum_{r=1}^{2d} \, \int_{ y\in\R}\Bigg(\frac{D_r\rho_{s-}D_r^*}{\Tr(D_r\rho_{s-}D_r^*)}-\rho_{s-}\Bigg)\ind_{0<y<\Tr(D_r\rho_{s-}D_r^*)} N^r(\der y, \der s)\, ,\\
 \der X_s&=&\sum_{r=1}^{2d} \, \int_{ y\in\R}e_r \, \ind_{0<y<\Tr(D_r^{}\rho_{s-}D_r^*)} N^r(\der y, \der s)  
 \end{eqnarray}
 and  $(\rho_0,X_0)=\left(\frac{\rho^{(0)}(i)}{\Tr\big(\rho^{(0)}(i)\big)},i\right)$ with probability $\Tr\big(\rho^{(0)}(i)\big)$.
\end{prop}

\noindent\textbf{Remark:} The second expression of the description of quantum trajectories is the exact continuous time analogue of the one described in \cite{MR3296637} for OQWs. Let us briefly explain how the quantum trajectories evolve in time. To this end we introduce:
\begin{equation}\label{jump}\forall r\in\{1,...,2d\}, \ \tilde{N}^r(t)=\int_{0}^{t}\int_{\R}\ind_{0<y<\Tr(D_r\rho_{s-}D_r^*)} N^r(\der y, \der s) \, .
\end{equation}
The processes $\tilde{N}^r$ are Poisson processes with intensity $\int_{0}^{t}\Tr(D_r\rho_{s-}D_r^*)\der s$. In particular the processes
\[\tilde{N}^r(t)-\int_{0}^{t}\Tr(D_r\rho_{s-}D_r^*)\der s\]
are martingales with respect to the filtration induced by $(\rho_t,X_t)_{t\geq 0}$. The evolution described by \eqref{qtcc} is deterministic and interrupted by jumps occurring at random time, it is typically a Piecewise Deterministic Markov Process. The jumps are generated by the Poisson processes \eqref{jump}.  As we can check from Eq. \eqref{qtcc}, if $\omega^{(0)}=\rho\otimes\vert i\rangle\langle i\vert$ for some $\rho\in\mathcal S_{\mathcal H}$ and $i\in\mathbb Z^d$ (that is $\vert X_0\rangle=\vert i\rangle$), the deterministic evolution let the position unchanged until a jump occurs. Since the Poisson processes $N^r$ are indepedent, only one Poisson process is involved. If $T$ denotes the time of the first jump and assume the process $N^r$ is involved, the internal degree of freedom is updated by $\rho_{T}=\frac{D_r\rho_{T-}D_r^*}{\Tr(D_r\rho_{T-}D_r^*)}$  and the position is changed and becomes $\vert i+e_r\rangle$. This means that the particle has jumped from the position $\vert i\rangle$ to the position $\vert i+e_r\rangle$. In other words we have $\vert X_t\rangle=\vert i\rangle$, for all $0\leq t\leq T-$ and $\vert X_T\rangle=\vert i+e_r\rangle$. Next, the deterministic evolution starts again with the new initial condition $\rho_T\otimes\vert i+e_r\rangle\langle i+e_r\vert$ until a new jump occur and so on.

%
 
%

The following result allows us to make the connection between CTOQWs and their associated quantum trajectories.

\begin{prop}\label{lien}
Let $\mu^{(t)}$ the OQW defined in Proposition \ref{proppel} and $\omega^{(t)}$ the associated quantum trajectory defined in Proposition \ref{omega}. Then we have
\[\forall t\geq 0, \ \E(\omega^{(t)})=\mu^{(t)} \, . \]

Moreover, for all $t\geq0$, the random variables $X_t$ and $Q_t$ have the same distributions $q_t$.
\end{prop}

\begin{proof}
The first part is proved in \cite{pellegrini2014continuous}. For the second part, let $\phi$ a bounded continuous map on $\Z^d$, we get:
\begin{align*}\E(\phi(Q_t))&=\sum\limits_{i\in\Z^d}\phi(i)\Tr\big(\mu^{(t)}(I\otimes|i\rangle\langle i|)\big)\\&=\sum\limits_{i\in\Z^d}\phi(i)\Tr\big(\E(\omega^{(t)})(I\otimes|i\rangle\langle i|)\big)\\
&=\sum\limits_{i\in\Z^d}\phi(i)\E\Bigg(\Tr\big(\omega^{(t)}(I\otimes|i\rangle\langle i|)\big)\Bigg)\\
&=\sum\limits_{i\in\Z^d}\phi(i)\E\big(\Tr(|X_t\rangle\langle X_t||i\rangle\langle i|)\big)\\
&=\sum\limits_{i\in\Z^d}\phi(i)\E(\ind_{X_t=i})\\
&= \E(\phi(X_t))
\, ,\end{align*}
and the result holds.
\end{proof}

In the next section, we state the CLT.


\section{Central Limit Theorem }\label{parttclc}

This section is devoted to prove the Central Limit Theorem for CTOQWs. The result holds under some assumption concerning the Lindblad operator on $\Hb$. This operator is defined below.
$$\begin{array}{llll}\mathcal L:&\Bb(\Hb)&\rightarrow&\Bb(\Hb)\\
&\rho&\mapsto & D_0 \rho+\rho D_0^*+\displaystyle{\sum_{r=1}^{2d}}D_r\rho D_r^* \, .
\end{array}$$
Our main assumption for the CLT is the following.
\begin{itemize}
\item $(H1)$ There exists a unique density matrix $\rho_{inv}\in\mathcal S_{\mathcal H}$ such that 
$$\mathcal L(\rho_{inv})=0 \, .$$
In particular $\dim\textrm{Ker}(\mathcal L)=1$.
\end{itemize}

Under the condition $(H1)$, we have the following ergodic theorem which is a particular case of the Ergodic Theorem of \cite{MR2106330}. In particular this theorem shall be useful in the proof of the CLT.

\begin{thm}[\cite{MR2106330}]\label{ergod}
	Assume $(H1)$. Let $(\rho_t,X_t)_{t\geq 0}$ the Markov process defined in Proposition \ref{omega}, therefore \[\dfrac{1}{t}\int_{0}^{t}\rho_{s}\der s\overset{a.s.}\longrightarrow \rho_{inv} \, .\]
\end{thm}

 Now, our strategy to show the CLT consists in reducing the problem to a CLT for martingales with the help of the solution of the Poisson equation. To this end let us introduce the generator of the process $(\rho_t,X_t)_{t\geq 0}$. 

We denote $\A$ the Markov generator of the process $(\rho_t,X_t)_{t\geq 0}$ and $\mathbb{D}(\A)$ its domain. For all $f\in\mathbb{D}(\A)$, $\rho\in\mathcal S_{\mathcal H}$ and $x\in\mathbb Z^d$, we get 
\begin{eqnarray}
 \A f(\rho,x)&=&\Der_{\rho}f(\F(\rho))\nonumber\\ &&+\sum_{r=1}^{2d}\left[f\left(\frac{D_r\rho D_r^*}{\Tr(D_r\rho D_r^*)},x+e_r\right)-f(\rho,x)\right]\,\Tr(D_r\rho D_r^*)  
\end{eqnarray}
where $\F(\rho)=D_0\rho+\rho  D_0^*-\rho\Tr(D_0\rho+\rho D_0^*)$ for all $\rho\in\mathcal S_{\mathcal H}$ and where $\Der_{\rho}f$ denotes the partial differential of $f$ with respect to $\rho$. \\

\noindent\textbf{Remark:} Note that in the sequel we do not need to make precise the exact domain of $\A$. Actually we shall apply the Markov generator on $\mathcal C^1$ functions.\\

%

We shall also need the following quantity,
$$m=\sum_{r=1}^{2d}\Tr(D_r\rho_{inv} D_r^*)e_r \, .$$
The following lemma shall be used in the proof.

\begin{lem}
	For all $u\in \R^d$, the equation
	\begin{equation}\label{1lemd}
	\mathcal L^*(J_u)=-\Big(\sum_{r=1}^{2d}(e_r.u)D_r^*D_r-(m.u)I\Big)
	\end{equation}
	admits a solution and the difference between any couple of solutions of (\ref{1lemd}) is a multiple of the identity. 
	\end{lem}
	
	\begin{proof}
		
	First, let us remark that
			\[\Tr\left(\rho_{inv}\left(\sum_{r=1}^{2d}(e_r.u)D_r^*D_r-(m.u)I\right)\right)=\sum_{r=1}^{2d}\Tr(D_r\rho_{inv} D_r^*)(e_r.u) - (m.u) \Tr(\rho_{inv})=0 \, ,\]
	which implies that $-\left(\displaystyle{\sum_{r=1}^{2d}}(e_r.u)D_r^*D_r-(m.u)I\right)\in\{\rho_{inv}\}^\perp$.
			But by hypothesis, we have $\{\rho_{inv}\}^\perp=\ke(\mathcal L)^\perp$. Moreover, since $\ke(\mathcal L)^\perp=\im(\mathcal L^*)$, we finally get that
			\[-\left(\sum_{r=1}^{2d}(e_r.u)D_r^*D_r-(m.u)I\right)\in\im(\mathcal L^*)\]
			which proves the existence of the lemma. Now we prove the second part. To this end consider $J_u$ and $J_u'$ two solutions of (\ref{1lemd}) and set $H_u=J_u-J_u'$. It is then clear that
			\[\mathcal L^*(H_u)=0 \, .\]
			Therefore $H_u\in\ke(\mathcal L^*)$. Since $\dim\ke(\mathcal L)=1$, we get $\dim\ke(\mathcal L^*)=1$ and since $\mathcal L^*(I)=D_0^*+D_0+\sum_{r=1}^{2d}D_r^*D_r=0$, the operator $H_u$ is necessarily a multiple of the identity.	
		\end{proof}
	
From now on, for $u\in\R^d$, we denote $J_u$ the unique solution of (\ref{1lemd}) such that $\Tr(J_u)=0$. Moreover, if $u=e_r$, then we simply write $J_u=J_r$. Using the linearity of $\mathcal L^*$, one can notice that:
	\[J_u=\sum_{r=1}^{d}u_r J_r \,, \]
	for all $u=(u_1,\ldots,u_d)\in\mathbb R^d.$
	
	The next lemma concerns the Poisson equation in our context (see \cite{MR1887206} for more details on the Poisson equation). 
	
	\begin{lem}\label{lempoid} For all $(\rho,x)\in\mathbb{S}\times\Z^d$ and $u\in\R^d$, let set
	\begin{equation}
			f_u(\rho,x)= \Tr(\rho J_u)+x.u \, .
			\end{equation}
			Then $f_u$ is solution of the Poisson equation:
	\begin{equation}\label{poid}
	\A f_u(\rho,x)=m.u \, .
	\end{equation}

	\end{lem}
	\begin{proof}
		For all $(\rho,x)\in\mathbb{S}\times\Z^d$ and $u\in\R^d$, we complete the following computation:
		\[\begin{array}{lll}
		\A f_u(\rho,x)&=& \Tr(\F(\rho) J_u)\\
		\\
		&&+\sum\limits_{r=1}^{2d}\left[\Tr\Bigg(\frac{D_r\rho D_r^*}{\Tr(D_r\rho D_r^*)}J_u\Bigg)+x.u+e_r.u-\Tr(\rho J_u)-x.u\right]\Tr(D_r\rho D_r^*)\\
		\\
		&=&\Tr\Big(D_0\rho J_u+\rho D_0^* J_u-\Tr(D_0\rho+\rho D_0^*)\rho J_u\Big)\\
		&&+\sum\limits_{r=1}^{2d}\Tr({D_r\rho D_r^*}J_u)+\Tr(D_r\rho D_r^*)(e_r.u)-\Tr\Big(\Tr(D_r\rho D_r^*)\rho J_u\Big)\\
		&=&\Tr\Bigg(\rho\Bigg[J_u D_0+D_0^* J_u+\sum\limits_{r=1}^{2d}D_r^*J_uD_r+\sum\limits_{r=1}^{2d}D_r^* D_r(e_r.u)\Bigg]\Bigg)\\
		&=&\Tr\Bigg(\rho\Bigg[\mathcal L^*(J_u)+\sum\limits_{r=1}^{2d}D_r^* D_r(e_r.u)\Bigg]\Bigg)\\
		&=&\Tr\Big((m.u)\rho\Big)\\
		&=& m.u \, ,
		\end{array}
	\]
		so $f_u$ is solution of the Poisson equation (\ref{poid}).
		\end{proof}
		
		Now we have found the solution of the Poisson equation, we express the CLT for martingales that we shall use.
		
		\begin{thm}[\cite{crimaldi2005convergence}]\label{tclmd}
	Let $(M_t)_{t\geq0}$ be a real, càdlàg, and square integrable martingale. Suppose the following conditions:
			\begin{equation}\label{eqt1d}
			\lim\limits_{t \to \infty}\E\left(\frac{1}{\sqrt{t}}\sup_{0\leq s \leq t}|\Delta M_s|\right)=0
			\end{equation}
			and
			\begin{equation}\label{eqt2d}
			\lim\limits_{t \to \infty}\dfrac{[M,M]_t}{t}=\sigma^2
			\end{equation}
			for some $\sigma\geq 0$, then
			\[\dfrac{M_t}{\sqrt{t}} \underset{t \to +\infty}{\overset{\Lin}{\longrightarrow}}\mathcal{N}(0,\sigma^2) \, . \]
	\end{thm}
	
	We shall also use the following lemma which is a straightforward consequence of the law of large numbers for martingales (see \cite{MR1725357}).
	
 	\begin{lem}\label{lfgnmd}
 	Let $Z_t$ a  real, càdlàg, and square integrable martingale which satisfies $\langle Z,Z\rangle_t \leq Kt$ for a constant $K$, then \[\frac{Z_t}{t}\overset{a.s.}\longrightarrow 0 \, .\]
 	\end{lem}

The last lemma below shall be useful in this part as well as in the next one. From now on, we denote $|u|$ the Euclidean norm of $u\in\R^d$. 
\begin{lem}\label{majo}For 	all $t\geq0$ and all $u\in\R^d$, we have
\begin{align*}&\E\left[\sup_{0\leq s\leq t}|X_{s}-X_0|\right]\leq (2d)t\quad \mathrm{and}\\&\E\left[\sup_{0\leq s\leq t}e^{u.(X_{s}-X_0)}\right]\leq \exp\Big((2d)t(e^{|u|}-1)\Big)\, .
\end{align*}
\end{lem}
	\begin{proof}
	Let $t\geq0$ and $u\in\R^d$,
	\[\begin{array}{lll}
	\E\left[\sup_{0\leq s\leq t}e^{u.(X_{s}-X_0)}\right]&\leq&\E\Big[e^{|u|\int_{0}^{t}|\der X_s|}\Big]\\
	&\leq&\E\left[\exp\left(|u|\sum_{r=1}^{2d} \, \int_{ v=0}^t\int_{ y\in\R} \ind_{0<y<\Tr(D_r^{}\rho_{v}D_r^*)} N^r(\der y, \der v)\right)\right] \\
	\\
	
	&\leq& \E\left[\exp\left(|u|\sum_{r=1}^{2d} \, \int_{ 0}^t\int_{ y=0}^1  N^r(\der y, \der v)\right)\right] \, .
	\end{array}\]
	Since $N^r$ are independent Poisson point processes on $\R^2$, we get
	\[\begin{array}{lll}
\E\left[\sup_{0\leq s\leq t}e^{u.(X_{s}-X_0)}\right] &\leq&\E\left[\exp\left(|u| \int_{ 0}^t\int_{ y=0}^1\,  N^1(\der y, \der v)\right)\right]^{2d}\\
	&\leq&\exp\Big((2d)t(e^{|u|}-1)\Big) \, .
	
	\end{array}\]
	In the same way, one can prove that \[\E\left[\sup_{0\leq s\leq t}|X_{s}-X_0|\right]\leq\E\left[\sum_{r=1}^{2d} \, \int_{ v=0}^t\int_{ y=0}^1  N^r(\der y, \der v)\right]\leq (2d)t \, .\] \end{proof}
	
	Now, we are in the position to state the main result of this section. 
	\clearpage	
	\begin{thm}\label{tcld}
	Assume $(H_1)$ holds. Let $(\rho_t,X_t)_{t\geq 0}$ the Markov process defined in Proposition \ref{omega} then \[\dfrac{X_t-tm}{\sqrt{t}} \underset{t \to +\infty}{\overset{\Lin}{\longrightarrow}}\mathcal{N}(0,V) \, , \]
		where $V\in\mathcal{M}_d(\R)$ such that for all $r,q\in\{1,...,d\}$, \[\begin{array}{lll}V_{rq}&=&-m_q\Tr(\rho_{inv}J_r)-m_r\Tr(\rho_{inv}J_q)\\
		&&+\delta_{rq}\big(\Tr(D_r\rho_{inv}D_r^*)+\Tr(D_{r+d}\rho_{inv}D_{r+d}^*)\big)\\
		&&+\Tr(D_q\rho_{inv}D_q^*J_r)+\Tr(D_r\rho_{inv}D_r^*J_q)\\
		&&-\Tr(D_{q+d}\rho_{inv}D_{q+d}^*J_r)-\Tr(D_{r+d}\rho_{inv}D_{r+d}^*J_q)\, .\end{array}\]
	\end{thm}
	
		\noindent\textbf{Remark:} Proposition \ref{lien} implies then the CLT for the process $(Q_t)_{t\geq0}$ as it holds for $(X_t)_{t\geq0}$.
	\begin{proof}
	As announced, the proof is a combination of Lemma \ref{lempoid} and Theorem \ref{tclmd}. Let $u\in\R^d$ and $f_u$ the $\mathcal C^1$ function defined in Lemma \ref{lempoid}. Since $\mathcal A$ is the generator of $(\rho_t,X_t)_{t\geq0}$, following the theory of problem of martingale, the process $(M_t)_{t\geq0}$ defined by
\begin{eqnarray}M_t&=&f_u(\rho_t,X_t)-f_u(\rho_0,X_0)-\int_{0}^{t}\A f_u(\rho_{s-},X_{s-})\der s \nonumber\\&=&\Tr(\rho_t J_u)-\Tr(\rho_0 J_u)+ X_t.u-X_0.u-(m.u)t \nonumber
\end{eqnarray}
is a local martingale with respect to the filtration $\F$ associated to $(\rho_t,X_t)_{t\geq0}$  (see \cite{MR1725357,ethier1986markov} for more details on problem of martingale). In order to apply Theorem \ref{tclmd}, we shall show that $(M_t)$ is a true martingale. To this end it is sufficient to show that $\E\Bigg(\sup_{0\leq s\leq t}|M_s|\Bigg) <\infty$ (see \cite{ethier1986markov} for more details). This way, since $|\Tr(\rho J_u)|\leq\|J_u\|_\infty$ for all $ \rho\in\Sb_\Hb$, one can check with the help of Lemma \ref{majo} that 
\[\E\Bigg(\sup_{0\leq s\leq t}|M_s|\Bigg)\leq 2\|J_u\|_\infty+2d|u|t+|m.u|t  \, .\] 

Now we shall see that $(M_t)$ fulfills the conditions of Theorem \ref{tclmd}. The first one is the easiest one. Indeed,
	\[
	|\Delta M_s| \leq |\Tr(\Delta\rho_s J_u)|+|\Delta X_s.u|\leq 2\|J_u\|_\infty + |u| \, .
	\]
This shows that $\Delta M_s$ is bounded independently of $s$ and thus the condition (\ref{eqt1d}) holds. Now, we check that $(M_t)$ satisfies Equation (\ref{eqt2d}). The bracket $[M,M]_t$ satisfies:
	\begin{eqnarray*}
\der[M,M]_s
	&=&\der[u.X,u.X]_s+2 \, \der[u.X,\Tr(\rho J_u)]_s+\der[\Tr(\rho J_u),\Tr(\rho J_u)]_s\\
	\\
	&=& \sum\limits_{r=1}^{2d}(e_r.u)^2\tilde{N}^r(\der s)
	+2\sum\limits_{r=1}^{2d}(e_r.u)\Tr\Bigg(\frac{D_r\rho_{s-}D_r^*}{\Tr(D_r\rho_{s-}D_r^*)}J_u\Bigg)\tilde{N}^r(\der s)\\
	\\&&-2\sum\limits_{r=1}^{2d}(e_r.u)\Tr(\rho_{s-}J_u)\tilde{N}^r(\der s)+\sum\limits_{r=1}^{2d}\Tr\Bigg(\frac{D_r\rho_{s-}D_r^*}{\Tr(D_r\rho_{s-}D_r^*)}J_u\Bigg)^2\tilde{N}^r(\der s)\\
	\\
	&&-2\sum_{r=1}^{2d}\Tr\Bigg(\frac{D_r\rho_{s-}D_r^*}{\Tr(D_r\rho_{s-}D_r^*)}J_u\Bigg)\Tr(\rho_{s-}J_u)\tilde{N}^r(\der s)+\sum\limits_{r=1}^{2d}\Tr(\rho_{s-}J_u)^2\tilde{N}^r(\der s) \,\\
	\\
	\\
	&=&\sum_{r=1}^{2d}\Bigg(\Tr\Bigg(\frac{D_r\rho_{s-}D_r^*}{\Tr(D_r\rho_{s-}D_r^*)}J_u\Bigg)^2 -\Tr(\rho_{s-}J_u)^2\Bigg)\tilde{N}^r(\der s)
			\\	 
			 		 &&-2\Tr(\rho_{s-}J_u)\sum_{r=1}^{2d}\left(\Tr\Bigg(\frac{D_r\rho_{s-}D_r^*}{\Tr(D_r\rho_{s-}D_r^*)}J_u\Bigg)-\Tr(\rho_{s-}J_u)+(e_r.u) \,\right) \tilde{N}^r(\der s)
			 		 \\	 			 		 
			 &&+\sum_{r=1}^{2d}(e_r.u)^2\tilde{N}^r(\der s)	+2\sum_{r=1}^{2d}(e_r.u)\Tr\Bigg(\dfrac{D_r\rho_{s-}D_r^*}{\Tr(D_r\rho_{s-}D_r^*)}J_u\Bigg)\tilde{N}^r(\der s) \, .
	\end{eqnarray*}
	
%
	
Now we shall make the martingales $Y^r(t)=\tilde{N}^r(t)-\int_{0}^{t}\Tr(D_r\rho_{s-}D_r^*)\der s$ appear in the first and the last term of the above expression. Concerning the second term, we recognize $\der\Tr(\rho_{s}J_u)$ and $\der (X_s.u)$ to get

\begin{eqnarray}\label{eee}	\der[M,M]_s&=
			&\sum_{r=1}^{2d}\left(\Tr\left(\frac{D_r\rho_{s-}D_r^*}{\Tr(D_r\rho_{s-}D_r^*)}J_u\right)^2 -\Tr(\rho_{s-}J_u)^2\right) Y^r(\der s)
				 		 \nonumber\\
		 &&+\sum_{r=1}^{2d}\left(\Tr\left(\frac{D_r\rho_{s-}D_r^*}{\Tr(D_r\rho_{s-}D_r^*)}J_u\right)^2 -\Tr(\rho_{s-}J_u)^2\right)\Tr(D_r\rho_{s-}D_r^*)\der s\
				 		 				\nonumber \\
		  &&-2\Tr(\rho_{s-}J_u)\Big(\der\Tr(\rho_{s}J_u)+\der (X_s.u)-\Tr(\F(\rho_{s-})J_u)\der s\Big)\nonumber\\
		 &&+\sum_{r=1}^{2d}(e_r.u)^2 Y^r(\der s) 	+2\sum_{r=1}^{2d}(e_r.u)\Tr\left(\frac{D_r\rho_{s-}D_r^*}{\Tr(D_r\rho_{s-}D_r^*)}J_u\right) Y^r(\der s)\nonumber\\
				 &&+\sum_{r=1}^{2d}(e_r.u)^2 \Tr(D_r\rho_{s-}D_r^*)\der s+	2\sum_{r=1}^{2d}(e_r.u)\Tr({D_r\rho_{s-}D_r^*}J_u) \der s \, .
\end{eqnarray}

One can remark that for $h(\rho,x):=\Tr(\rho J_u)^2$, we get for all $(\rho,x)\in\Sb_\Hb\times\Z^d$:
\begin{align}\label{genh}
\A  h(\rho,x)=& \, 2\Tr\big(\F(\rho)J_u\big)\Tr\big(\rho J_u\big) \nonumber\\
&+\sum_{r=1}^{2d}\left(\Tr\left(\frac{D_r\rho D_r^*}{\Tr(D_r\rho D_r^*)}J_u\right)^2 -\Tr(\rho J_u)^2\right)\Tr(D_r\rho D_r^*)
\, .\end{align}
Since $h\in\mathcal C^1$, the process $(S_t^h)_{t\geq 0}$ defined by:
\[S_t^h=\Tr(\rho_t J_u)^2-\Tr(\rho_0 J_u)^2-\int_0^t \A  h(\rho_{s-},X_{s-})\der s \] 
is a local martingale. Besides, since $|\Tr(\rho J_u)|\leq\|J_u\|_\infty$ for all $ \rho\in\Sb_\Hb$, one has
 \[\E\Bigg(\sup_{0\leq s\leq t}|S_t^h|\Bigg)\leq\alpha+\beta t <\infty\]
 for all $t\geq 0$, so $S_t^h$ is actually a true martingale.

 Now using Equation (\ref{genh}) in the second line of (\ref{eee}) and recognizing $\der M_t$ in the third one, we have
\begin{eqnarray}\label{brabra}	\der[M,M]_s&=
	&\sum_{r=1}^{2d}\Bigg(\Tr\Bigg(\frac{D_r\rho_{s-}D_r^*}{\Tr(D_r\rho_{s-}D_r^*)}J_u\Bigg)^2 -\Tr(\rho_{s-}J_u)^2\Bigg) Y^r(\der s)\nonumber\\
	&&+ \A  h(\rho_{s-},X_{s-})\der s-2\Tr\big(\F(\rho_{s-})J_u\big)\Tr\big(\rho_{s-}J_u\big)\der s\nonumber\\
	&&-2\Tr(\rho_{s-}J_u)\Big(\der M_s +(m.u)\der s-\Tr\big(\F(\rho_{s-})J_u\big)\der s\Big)\nonumber\\
	&&+\sum_{r=1}^{2d}(e_r.u)^2 Y^r(\der s) +	+2\sum_{r=1}^{2d}(e_r.u)\Tr\Bigg(\frac{D_r\rho_{s-}D_r^*}{\Tr(D_r\rho_{s-}D_r^*)}J_u\Bigg) Y^r(\der s)\nonumber\\
	&&+\sum_{r=1}^{2d}(e_r.u)^2 \Tr(D_r\rho_{s-}D_r^*)\der s	+2\sum_{r=1}^{2d}(e_r.u)\Tr({D_r\rho_{s-}D_r^*}J_u) \der s \nonumber\\
	&=&
 		\sum\limits_{r=1}^{2d}\Bigg(\Tr\Bigg(\frac{D_r\rho_{s-}D_r^*}{\Tr(D_r\rho_{s-}D_r^*)}J_u\Bigg)^2 -\Tr(\rho_{s-}J_u)^2+(e_r.u)^2\nonumber\\&&\hphantom{ccccc}+2(e_r.u)\Tr\Bigg(\frac{D_r\rho_{s-}D_r^*}{\Tr(D_r\rho_{s-}D_r^*)}J_u\Bigg)\Bigg)Y^r(\der s)\nonumber\\
 		&&-\der S_s^h-2\Tr(\rho_{s-}J_u)\der M_s+\der\big(\Tr(\rho_s J_u)^2\big)\nonumber\\
 		&&+\Tr\left(\rho_{s-}\left[-2(m.u)J_u+	\sum\limits_{r=1}^{2d}(e_r.u)^2D_r^*D_r+2\sum_{r=1}^{2d}(e_r.u)D_r^*J_u D_r\right]\right)\der s \, .\nonumber\\&&
 		\end{eqnarray}
Let us denote by $H_s^r$ the term in front of $ Y^r(\der s)$. Now, we shall apply Lemma \ref{lfgnmd}. To this end, recall that
%
%
for all $ \rho\in\Sb_\Hb, \,  |\Tr(\rho J_u)|\leq\|J_u\|_\infty$,  this implies the following estimates:
 		\[
 	\langle\int_0^.H_s^r\der Y^r_{s},\int_0^.H_s^r\der Y^r_{s}\rangle_t\leq(2\|J_u\|_\infty^2 +|u|^2+2|u|\|J_u\|_\infty)^2\|D_r^*D_r\|_\infty t \, ,\]
 		\[\langle\int_0^.-2\Tr(\rho_{s-}J_u)\der M_{s},\int_0^.-2\Tr(\rho_{s-}J_u)\der M_{s}\rangle_t\leq4\|J_u\|_\infty^2(|u|+2\|J_u\|_\infty)^2\left(\sum_{r=1}^{2d}\|D_r^*D_r\|_\infty\right)t \, ,\]
 		 		\[\langle S^h,S^h\rangle_t\leq 64\|J_u\|_\infty^4\left(\sum_{r=1}^{2d}\|D_r^*D_r\|_\infty\right)t \, .\]
Lemma \ref{lfgnmd} shows that only the last term of \eqref{brabra} contributes to $\lim\limits_{t \to \infty}\dfrac{[M,M]_t}{t}$. Applying Theorem \ref{ergod}, we get
 		
 		\begin{align*}\lim\limits_{t \to \infty}\dfrac{[M,M]_t}{t}&=\lim\limits_{t \to \infty}\frac{1}{t}\int_{0}^t \Tr\Bigg(\rho_{s-}\Bigg[-2(m.u)J_u+	\sum\limits_{r=1}^{2d}(e_r.u)^2D_r^*D_r\\
 		&\hphantom{ccccccccccccccclcccccccccccccccccc}+2\sum\limits_{r=1}^{2d}(e_r.u)D_r^*J_u D_r\Bigg]\Bigg)\der s\\
 		&=\Tr\left(\rho_{inv}\left[-2(m.u)J_u+	\sum_{r=1}^{2d}(e_r.u)^2D_r^*D_r+2\sum_{r=1}^{2d}(e_r.u)D_r^*J_u D_r\right]\right) \, .\end{align*}
 		
 		Now defining $\sigma^2_u=\Tr\left(\rho_{inv}\left[-2(m.u)J_u+	\sum\limits_{r=1}^{2d}(e_r.u)^2D_r^*D_r+2\sum\limits_{r=1}^{2d}(e_r.u)D_r^*J_u D_r\right]\right)$, Theorem \ref{tclmd} states that:
 			\[\dfrac{M_t}{\sqrt{t}}=\dfrac{X_t.u-(m.u)t+\Tr(\rho_t J_u)-\Tr(\rho_0 J_u)-X_0.u}{\sqrt{t}} \underset{t \to +\infty}{\overset{\Lin}{\longrightarrow}}\mathcal{N}(0,\sigma^2_u) \, . \]

 Since $\Big(\Tr(\rho_t J_u)-\Tr(\rho_0 J_u)-X_0.u\Big)$ is bounded independently of $t$, one can obviously deduce that for all $u=(u_1,...,u_d)\in\R^d$, one has
 			\[\dfrac{X_t.u-(m.u)t}{\sqrt{t}} \underset{t \to +\infty}{\overset{\Lin}{\longrightarrow}}\mathcal{N}(0,\sigma^2_u),  \]
 		where \[\begin{array}{lll}\sigma^2_u&=&-2\sum\limits_{r,q=1}^{d}u_ru_qm_q\Tr(\rho_{inv}J_r)
 		+\sum\limits_{r=1}^{d}u_r^2\Big(\Tr(D_r\rho_{inv}D_r^*)+\Tr(D_{r+d}\rho_{inv}D_{r+d}^*)\Big)\\&&+2\sum\limits_{r,q=1}^{d}u_ru_q\Big(\Tr(D_q\rho_{inv}D_q^*J_r)-\Tr(D_{q+d}\rho_{inv}D_{q+d}^*J_r)\Big) \, .\end{array}\]
		 Finally we can check that $\sigma^2_u=\sum\limits_{r,q=1}^{d}u_r u_q	V_{rq}$ for all $u=(u_1,...,u_d)\in\R^d$, which ends the proof.
	\end{proof}
	
We finish this section by specifying the case $d=1$. This is the simpler case where the walker can only jump to the right or the left. The Markov process $(\rho_t,X_t)_{t\geq 0}$, with values in $\mathcal{S_{\Hb}}\times\Z$, is defined by the following differential equations:
	\begin{equation*}\begin{array}{llll}
	
	 &\der\rho_{s}&=& \Big(D_0\rho_{s-}+\rho_{s-}D_0^*-\rho_{s-}\Tr(D_0\rho_{s-}+\rho_{s-}D_0^*)\Big)\der s \\
	 \\
	 &&& + \Big(\frac{D_1 \rho_{s-}D_1^*}{\Tr(D_1\rho_{s-}D_1^*)}-\rho_{s-}\Big) \tilde{N}^1(\der s)+ \Big(\frac{D_2 \rho_{s-}D_2^*}{\Tr(D_2\rho_{s-}D_2^*)}-\rho_{s-}\Big) \tilde{N}^2(\der s)\\ \\
	 \text{and}&\der X_s&=&\tilde{N}^1(\der s)-\tilde{N}^2(\der s),
	 \end{array}
	 \end{equation*}
	 where $D_0,D_1,D_{2}\in \Bb(\Hb)$ such that $D_0+D_0^*+D_1^* D_1+D_{2}^* D_{2}=0$.
	
	\begin{thm}\label{tcl} Suppose that the Lindblad operator $$\mathcal L(\rho)=D_0\rho+\rho D_0^*+D_1\rho D_1^* + D_{2}\rho D_{2}^*$$ admits a unique density matrix $\rho_{inv}$ such that $\mathcal L(\rho_{inv})=0$.\\
	
	Set $m=\Tr(D_1\rho_{inv} D_1^*)-\Tr(D_{2}\rho_{inv} D_{2}^*)$, and let $J$ the unique solution of $\mathcal L^*(J)=-D_1^* D_1 + D_{2}^* D_{2}+m I$ such that $Tr(J)=0$.\\
	 
	Then, we have the following CLT
	
			 \[\dfrac{X_t-tm}{\sqrt{t}} \underset{t \to +\infty}{\overset{\Lin}{\longrightarrow}}\mathcal{N}(0,\sigma^2) \]
			where $\sigma^2=\Tr(\rho_{inv}[-2mJ+D_1^*D_1+D_{2}^*D_{2}+2D_1^*J D_1-2D_{2}^*JD_{2}])$.
		\end{thm}

\section{Large Deviation Principle}

Here, we study a Large Deviation Principle (LDP) for CTOQWs. Our proof is inspired by strategies developed in \cite{jakvsic2014entropic,carbone2014homogeneous} which are essentially based on the application of the Gärtner-Ellis Theorem (\cite{dembo2009large}). 
 In the following, one can notice that the Perron-Frobenius Theorem for positive maps (\cite{evans1978spectral}) is the main tool to apply the Gärtner-Ellis Theorem.\\

In order to prove the LDP, we shall use a deformed Lindblad operator. From now on, we define for all $u\in\R^d$, the operators $D_r^{(u)}=e^\frac{u.e_r}{2}D_r$,  $r\in\{0,\ldots,2d\}$, and we denote $\mathcal L^{(u)}$ the deformed Lindblad operator associated to the operators $D_r^{(u)}$, that is,
\[ \, \mathcal L^{(u)}(\rho)=D_0\rho+\rho D_0^*+\sum\limits_{r=1}^{2d}e^{u.e_r}D_r\rho D_r^* \, ,\]
for all $\rho\in\mathcal{S_{\Hb}}.$

Now, defining
$$\phi^{(u)}(\rho)=\sum\limits_{r=1}^{2d}e^{u.e_r}D_r\rho D_r^*,$$
for all $\rho\in\mathcal S_{\mathcal H}$ and all $u\in\mathbb R^d$, we can see that $\mathcal L^{(u)}$ is written in the usual Lindblad form, that is $\mathcal L^{(u)}(\rho)=D_0 \rho+\rho D_0^*+\phi^{(u)}(\rho)$. This way, the semi-group $\{e^{t\mathcal L^{(u)}}\}_{t\geq 0}$ is a completely positive (CP) semi-group (see \cite{MR551466} for the proof). In a same way, we write $\Lin$ as:
\[\mathcal L(\rho)=D_0 \rho+\rho D_0^*+\phi(\rho) \]
for all $\rho\in\mathcal{S_{\Hb}}$ where $\phi(\rho)=\sum\limits_{r=1}^{2d}D_r\rho D_r^*$.\\

In this part, the notion of irreducibility is required. This notion was originally defined in \cite{MR0269241}. There are several equivalent definitions that the reader can find in \cite{carbone2014homogeneous,carbone2014open}.
\begin{defi}\label{irr}
The CP map $\phi:\rho\mapsto\sum\limits_{r=1}^{2d}D_r\rho D_r^*$ is called irreducible if for any non-zero $x\in\Hb$, the set $\C[D]x$ is dense in $\Hb$, where $\C[D]$ is the set of polynomials in $D_r$, $r\in\{1,...,2d\}$.
\end{defi}

In the sequel, we need $\{e^{t\mathcal L^{(u)}}\}_{t\geq 0}$ to be positivity improving. This means that for all $t>0$ and for all $A\geq0$, $A\in\Bb(\Hb)$, one has $e^{t\mathcal L^{(u)}}(A)>0$.  The following lemma provides an effective criterion for verifying that $\{e^{t\mathcal L^{(u)}}\}_{t\geq 0}$ is positivity improving.

\begin{lem}\label{piouf}
 If $\phi$ is irreducible, then $\phi^{(u)}$ is irreducible by a direct application of Definition \ref{irr} and therefore $\{e^{t\mathcal L^{(u)}}\}_{t\geq 0}$ is positivity improving (\cite{jakvsic2014entropic}).
\end{lem}

The next two lemmas are relevant in the proof of the main theorem of this part. In particular, the following lemma describes the largest eigenvalue associated to the deformed Lindblad semigroup $\{e^{t\mathcal L^{(u)}}\}_{t\geq 0}$.
\begin{lem}\label{perronf}Let $t\geq 0$,
suppose that $\phi$ is irreducible. Set 
\[l_u=\max\{\mathrm{Re}(\lambda),\lambda\in\Sp(\mathcal L^{(u)})\} \, .\] Then $e^{t l_u}$ is an algebraically simple
eigenvalue of $e^{t\mathcal L^{(u)}}$, and the associated eigenvector $V_u$ is strictly positive (which can be normalized to be in $ \mathcal{S_{\Hb}}$). Besides, the map $u\mapsto l_u$ can be extended to be analytic in a neighbourhood of $\R^d$.
\end{lem}

\begin{proof}
The first part has been proved in \cite{jakvsic2014entropic}, the proof is based on the Perron-Frobenius Theorem for CP maps (\cite{evans1978spectral}). In particular, using such result, one get that $e^{t l_u}$ is an geometric simple
eigenvalue of $e^{t\mathcal L^{(u)}}$, and the associated eigenvector $V_u\in \Bb(\Hb)$ is strictly positive. It remains to show the algebraic simplicity of the eigenvalue $e^{t l_u}$.  To this end, we introduce:

\[\forall X\in \Bb(\Hb), \Psi(X)=V_u^{-\frac{1}{2}}e^{t(\mathcal L^{(u)}-l_u)}(V_u^{\frac{1}{2}}XV_u^{\frac{1}{2}}) \, V_u^{-\frac{1}{2}} \, .\]
Note that $V_u^{-\frac{1}{2}}$ is well defined since $V_u$ is strictly positive. The irreducibility of $\phi$ involves the positivity improving of $\Psi$ (with Lemma \ref{piouf}), and one has that the diamond norm of $\Psi$ is equal to one since $\Psi(I)=I$. Therefore, Theorem 2.2 in \cite{jakvsic2014entropic} implies that 1 is a geometrically simple eigenvalue of $\Psi$, and this holds for $\Psi^*$ too. By applying Theorem 2.5. of \cite{evans1978spectral}, one get that the associated eigenvector $X_1$ of $\Psi^*$ is positive. Assume by contradiction that 1 is not algebraically simple for $\Psi^*$. Then the Jordan decomposition shows that there exists $X_2$ such that $\Psi^*(X_2)=X_1+X_2 $. Besides $\Psi^*$ is trace preserving since $\Psi(I)=I$, therefore $\Tr(X_1)=0$, then $X_1=0$ which is impossible. This implies that $e^{t l_u}$ is algebraically simple for $e^{t\mathcal L^{(u)}}$. The analyticity of $u\mapsto l_u$ is a simple application of perturbation theory for matrix eigenvalues (see Chapter II in \cite{kato2013perturbation}).
\end{proof}

The next lemma describes the link between the moment generating function of $\big(X_t-X_0\big)$ and the deformed Lindblad semigroup.
\begin{lem}\label{dfl}
 For all $t\geq 0$ and all $u\in\R^d$, one has 
\[\E\Big(e^{u.(X_t-X_0)}\Big)= \Tr\Big(e^{t\mathcal L^{(u)}}(\E[\rho_0])\Big) \, .\]
\end{lem}
\begin{proof}
The idea of the proof consists in rewriting $\E\Big(e^{u.(X_t-X_0)}\Big)$ with the help of a Dyson expansion. From now, we set $u\in\R^d$ and $f:(\rho,x)\mapsto e^{u.x}\in\mathcal C^1$. Since $\mathcal A$ is also the generator of $(\rho_t,X_t-X_0)_{t\geq0}$, the process $(M_t^f)_{t\geq0}$ defined by
\[M_t^f=f(\rho_t,X_t-X_0)-f(\rho_0,0)-\int_0^t \A f(\rho_{t_1-},X_{t_1-}-X_0) \der {t_1}\] 
is a local martingale. Due to Lemma \ref{majo}, one has the following upper bound. 
\begin{align*}
\E\Bigg(\sup_{0\leq s\leq t}|M^f_s|\Bigg)\leq & \, \exp\Big((2d)t(e^{|u|}-1)\Big)+1\\
&+(e^{|u|}+1)\left(\sum_{r=1}^{2d}\|D_r^*D_r\|_\infty\right)\exp\Big((2d)t(e^{|u|}-1)\Big) \, .\end{align*}
 Then $\E\Bigg(\sup_{0\leq s\leq t}|M_s^f|\Bigg)<\infty$, for all $t\geq 0$ which implies that $(M_t^f)_{t\geq0}$ is a true martingale.  
This leads to
 $$\E\Big(f(\rho_t,X_t-X_0)\Big)=\E\Big(f(\rho_0,0)\Big)+\E\Big(\int_0^t \A f(\rho_{t_1},X_{t_1}-X_0) \der {t_1}\Big).$$ This way, we can develop $\E\Big(e^{u.(X_t-X_0)}\Big)$. For all $t\geq0$,
\begin{align}\label{16}
\E\Big(e^{u.(X_t-X_0)}\Big)&=1+\E\Bigg(\int_0^t \sum_{r=1}^{2d}\Big[e^{u.(X_{t_1}-X_0+e_r)}-e^{u.(X_{t_1}-X_0)}\Big]\Tr(D_r\rho_{t_1} D_r^*) \,  \der {t_1}\Bigg)\nonumber\\
&= 1+\E\Bigg(\int_0^t e^{u.(X_{t_1}-X_0)}\Bigg[\Tr(\mathcal L^{(u)}(\rho_{t_1}))\nonumber\\
&\hphantom{= \E\Big(\int_0^t e^{u.(llllllll)}[\Tr(}-\Tr(D_0\rho_{t_1}+\rho_{t_1}D_0^*+\sum_{r=1}^{2d}D_r\rho_{t_1} D_r^*)\Bigg] \,  \der {t_1}\Big)\nonumber\\
&= 1+\E\Bigg(\int_0^t e^{u.(X_{t_1}-X_0)}\Tr\Big(\mathcal L^{(u)}(\rho_{t_1})\Big) \,  \der {t_1}\Bigg)\nonumber\\
&= 1+\int_0^t \Tr\Bigg(\mathcal L^{(u)}\Bigg(\E\Big[e^{u.(X_{t_1}-X_0)}\rho_{t_1}\Big]\Bigg)\Bigg) \,  \der {t_1} \, .
\end{align}
In a similar way, we want to develop $\E\Big[e^{u.(X_{t_1}-X_0)}\rho_{t_1}\Big]$. Let $g:(\rho,x)\mapsto e^{u.x}\rho\in\mathcal C^1$, the process $(M^g_t)_{t\geq0}$ defined by
\[M_t^g=g(\rho_t,X_t-X_0)-g(\rho_0,0)-\int_0^t \A g(\rho_{t_2-},X_{t_2-}-X_0) \der {t_2}\] 
is a local martingale. One can also check that $\E\Bigg(\sup_{0\leq s\leq t}|M_s^g|\Bigg)<\infty$ for all $t\geq 0$, which implies that $(M_t^g)_{t\geq0}$ is a true martingale. And therefore, one has
\begin{eqnarray}\label{17}
\E\Big[e^{u.(X_{t_1}-X_0)}\rho_{t_1}\Big]&= &\E(\rho_0)+\E\Big(\int_{ 0}^{t_1} e^{u.(X_{t_2}-X_0)}\F(\rho_{t_2})\der t_2\Big)\nonumber\\
&&+\E\Bigg(\int_{ 0}^{t_1}\sum\limits_{r=1}^{2d}\Bigg[e^{u.(X_{t_2}-X_0+e_r)}\frac{D_r\rho_{t_2} D_r^*}{\Tr(D_r\rho_{t_2} D_r^*)}\nonumber\\
&&\hphantom{+\E\Bigg(\int_{ 0}^{t_1}\sum\limits_{r=1}^{2d}\Big[e^{u.(X_{t_2})}}-e^{u.(X_{t_2}-X_0)}\rho_{t_2}\Bigg]\Tr(D_r\rho_{t_2} D_r^*)   \der t_2\Bigg)\nonumber\\
&=&\E(\rho_0)+\E\Bigg(\int_{ 0}^{t_1} e^{u.(X_{t_2}-X_0)}\Bigg[\Tr\Big(\mathcal L^{(u)}(\rho_{t_2})\Big)\nonumber\\&&\phantom{ccccccccc}-\rho_{t_2}\Tr\Bigg(D_0\rho_{t_2}+\rho_{t_2}D_0^*+\sum\limits_{r=1}^{2d}D_r\rho_{t_2} D_r^*\Bigg)\Bigg]\der t_2\Bigg)\nonumber\\
&=&\E(\rho_0)+\int_{ 0}^{t_1}  \Tr\Bigg(\mathcal L^{(u)}\Big(\E\Big[e^{u.(X_{t_2}-X_0)}\rho_{t_2}\Big]\Big)\Bigg) \,  \der {t_2} \, .
\end{eqnarray}

We plug (\ref{17}) into (\ref{16}) and we get, for all $t\geq0$,
\begin{align*}\E\Big(e^{u.(X_t-X_0)}\Big)=&  \, 1+t\Tr\Big(\mathcal L^{(u)}(\E[\rho_{0}])\Big)\\
&+\int_0^t \int_0^{t_1} \Tr\Bigg(\left(\mathcal L^{(u)}\right)^2\Big(\E\Big[e^{u.(X_{t_2}-X_0)}\rho_{t_2}\Big]\Big)\Bigg) \,  \der {t_2}\der {t_1} \, .
\end{align*}
By iterating this procedure, we obtain 
\begin{align*}
\E\Big(e^{u.(X_t-X_0)}\Big)=& \Tr\Bigg(\sum_{k=0}^{j}\frac{t^k}{k!}\left(\mathcal L^{(u)}\right)^k\Big(\E[\rho_{0}]\Big)\Bigg)\nonumber\\&+\int_{0<t_j<...<t}  \Tr\Bigg(\left(\mathcal L^{(u)}\right)^{j+1}\Big(\E\big[e^{u.(X_{t_{j+1}}-X_0)}\rho_{t_{j+1}}\big]\Big)\Bigg) \,  \der {t_{j+1}}...\der {t_1} \, .\nonumber
\end{align*}
for all $j\in\mathbb N$. Now it is obvious that the first term converges to $\Tr\Bigg(e^{t\mathcal L^{(u)}}\Big(\E[\rho_0]\Big)\Bigg)$ when $j$ goes to infinity. In order to conclude it remains to prove that the second terms converges to zero. Let us estimate its norm.

\begin{eqnarray}\left\|\int_{0<t_j<...<t}  \Tr\Bigg(\left(\mathcal L^{(u)}\right)^{j+1}\Big(\E\Big[e^{u.(X_{t_{j+1}}-X_0)}\rho_{t_{j+1}}\Big]\Big)\Bigg) \,  \der {t_{j+1}}...\der {t_1}\right\|_1\hfill\nonumber\\
\hfill\leq\frac{t^{j+1}}{(j+1)!}\left\|\mathcal L^{(u)}\right\|_1^{j+1}\sup_{0\leq s\leq t}\E\Big[e^{u.(X_{s}-X_0)}\Big] \, .\nonumber
\end{eqnarray}

Finally, thanks to Lemma \ref{majo} and Jensen's inequality,
\begin{eqnarray}\left\|\int_{0<t_j<...<t}  \Tr\Bigg(\left(\mathcal L^{(u)}\right)^{j+1}\Big(\E\Big[e^{u.(X_{t_{j+1}}-X_0)}\rho_{t_{j+1}}\Big]\Big)\Bigg) \,  \der {t_{j+1}}...\der {t_1}\right\|_1\hfill\nonumber\\
\hfill\leq\frac{t^{j+1}}{(j+1)!}\left\|\mathcal L^{(u)}\right\|_1^{j+1}e^{(2d)t(e^{|u|}-1)} \nonumber
\end{eqnarray}

which converges to 0 when $j$ goes to infinity.
\end{proof}

Now, we can state the main result of this part.

\begin{thm}\label{ldppp}
	Let $(\rho_t,X_t)_{t\geq 0}$ the Markov process defined in Proposition \ref{omega}. Assume that $\phi$ is irreducible. The process $\Big(\frac{X_t-X_0}{t}\Big)_{t\geq0} $ satisfies a Large Deviation
Principle with a good rate function $\Lambda^*$.

Explicitly there exists a lower semicontinuous mapping $ \Lambda^*: \R^d \mapsto [0, +\infty]$ with compact level sets $\{x | \Lambda^*(x) \leq \alpha\}$, such
that, for all open set $G$ and all closed set $F$ with $G \subset F \subset\R^d$, one has:
\begin{align*}
-\inf_{x\in G} \Lambda^*(x)&\leq \liminf_{t\rightarrow+\infty}\frac{1}{t}\log\Pro\Bigg(\frac{X_t-X_0}{t}\in G\Bigg)\\
\\  &\leq \limsup_{t\rightarrow+\infty}\frac{1}{t}\log\Pro\Bigg(\frac{X_t-X_0}{t}\in F\Bigg)\leq-\inf_{x\in F} \Lambda^*(x) \, .
\end{align*}

 Moreover, $\Lambda^*$ can be expressed explicitly,  
\[\Lambda^*:x\mapsto\sup_{u\in\R^d}(u.x-l_u) \]
where $l_u$ is defined in Lemma \ref{perronf}.

\end{thm}
	\noindent\textbf{Remark:} Moreover, if $\E(e^{u.X_0})<\infty$ then the LDP holds for $(X_t)_{t\geq0}$ and not only for $(X_t-X_0)_{t\geq0}$. In this case,  Proposition \ref{lien} allows us to have the LDP for $(Q_t)_{t\geq0}$.
\begin{proof}
	The main tool of the proof is the Gärtner-Ellis Theorem (GET) (see \cite{dembo2009large}
	). We focus on the moment generating function which is involved in the GET. Let $t\geq0 $ and $u\in\R^d$, Lemma \ref{dfl} implies that 
	\[\E(e^{u.(X_t-X_0)})=\Tr\Big(e^{t\mathcal L^{(u)}}\big(\E[\rho_0]\big)\Big)=\sum\limits_{i\in\Z^d} \Tr\Big(e^{t\mathcal L^{(u)}}\big(\rho^{(0)}(i)\big)\Big) \, .\]
	
	 Set $0<\epsilon<t$. Due to Lemma \ref{piouf}, $e^{\epsilon\mathcal L^{(u)}} $ has the property of positivity improving, therefore $e^{\epsilon\mathcal L^{(u)}}\big(\rho^{(0)}(i)\big)$ is strictly positive for all $i\in\Z^d$.\\
	  If we set $r_{u,i}=\inf\Big(\Sp\Big[e^{\epsilon\mathcal L^{(u)}}\big(\rho^{(0)}(i)\big)\Big]\Big)>0$ and $s_{u,i}=\frac{\Tr\Big(e^{\epsilon\mathcal L^{(u)}}\big(\rho^{(0)}(i)\big)\Big)}{\inf\Sp(V_u)}$ 
	   then $\Sp\Big[e^{\epsilon\mathcal L^{(u)}}\big(\rho^{(0)}(i)\big)-r_{u,i} V_u\Big]\subset\R^{+}$ and $\Sp\Big[s_{u,i} V_u - e^{\epsilon\mathcal L^{(u)}}\big(\rho^{(0)}(i)\big)\Big]\subset\R^{+}$, and thus
		\[r_{u,i} V_u\leq e^{\epsilon\mathcal L^{(u)}}\Big(\rho^{(0)}(i)\Big)\leq s_{u,i} V_u \, .\]
		Since $e^{(t-\epsilon)\mathcal L^{(u)}} $ preserves the positivity, we get 
			\[r_{u,i} e^{(t-\epsilon) l_u} V_u\leq e^{t\mathcal L^{(u)}}\Big(\rho^{(0)}(i)\Big)\leq s_{u,i} e^{(t-\epsilon) l_u} V_u \, .\]
			Taking the trace, Lemma \ref{dfl} yields
			\[e^{(t-\epsilon) l_u}\sum\limits_{i\in\Z^d} r_{u,i} \leq \E\big(e^{u.(X_t-X_0)}\big)\leq e^{(t-\epsilon) l_u}\sum\limits_{i\in\Z^d} s_{u,i}  \, .\]
			The sums are finite and positive, then we have 
			\begin{equation*}
			\lim_{t\rightarrow+\infty}\frac{1}{t} \log\Big(\E\big(e^{u.(X_t-X_0)}\big)\Big)=l_u \, .
		\end{equation*}

 Now define $\Lambda_t: u\mapsto \log\Big(\E\big(e^{u.\frac{X_t-X_0}{t}}\big)\Big)$ the logarithm of the moment generating function of $\frac{X_t-X_0}{t} $, and $\Lambda:u\mapsto l_u$. We have shown that 
	\[\lim_{t\rightarrow+\infty}\frac{1}{t}\Lambda_t(tu)=\Lambda(u) \, .\]
Since $\Lambda$ is analytic (Lemma \ref{perronf}), Gärtner-Ellis Theorem can be applied, this proves the LDP and the associated good rate function is \[\Lambda^*:x\mapsto\sup_{u\in\R^d}(u.x-l_u) \, .\]
\end{proof}
\section{Examples}\label{ex}
This section is devoted to the illustration of the CLT and LDP with concrete examples.

Let us first start with two examples in the case $d=1$. Next, we provide an example for $d=2$.

In the case $d=1$, we obtain an open quantum walk on $\Z$ where the walker can stand still some random time, jump to the right and jump to the left. These evolutions are respectly governed by $D_0, D_1$ and $D_2$. We must have: \[D_0+D_0^*+D_1^*D_1+D_2^*D_2=0 \, .\]
In the case $d=2$, we shall need five operators $(D_r)_{r\in\{0,...,4\}} $ satisfying 
\[D_0+D_0^*+\sum_{r=1}^4 D_r^*D_r=0 \, .\]
Recall that we need to check condition $(H_1)$ for getting the CLT (Theorem \ref{tcld}) and we need to check the condition of Definition \ref{irr} for obtaining LDP (Theorem \ref{ldppp}).
\begin{enumerate}

\item The constraint above is respected in this concrete example:
	\[D_0=-\frac{1}{2}I, \ D_1=\dfrac{1}{\sqrt{3}}\begin{pmatrix}
					1&1\\
					0&1
					\end{pmatrix} \text{ and }D_{2}=\dfrac{1}{\sqrt{3}}\begin{pmatrix}
					1&0\\
					-1&1
					\end{pmatrix} \, .\]

								This example falls within the scope of the CLT, in fact $(H1)$ is checked with
								$\rho_{inv}=\frac{1}{2}I$.
							We get
								\[m=\Tr(D_1\rho_{{inv}} D_1^*)-\Tr(D_{2}\rho_{{inv}} D_{2}^*)=0 \ ; \ 
							J=\dfrac{1}{6}\begin{pmatrix}
						-5& 2\\
						2&5
							\end{pmatrix} \ \text{and} \ \sigma^2=\frac{8}{9} \, .\]
							
							Then the CLT states that
						\[\dfrac{X_t}{\sqrt{t}} \underset{t \to +\infty}{\overset{\Lin}{\longrightarrow}}\mathcal{N}\left(0,\frac{8}{9}\right)  \, ,\]
							 as it is illustrated by Figure \ref{f1}.
							
							\begin{figure}
							\begin{center}
							 \includegraphics[width=140mm,height=60mm]{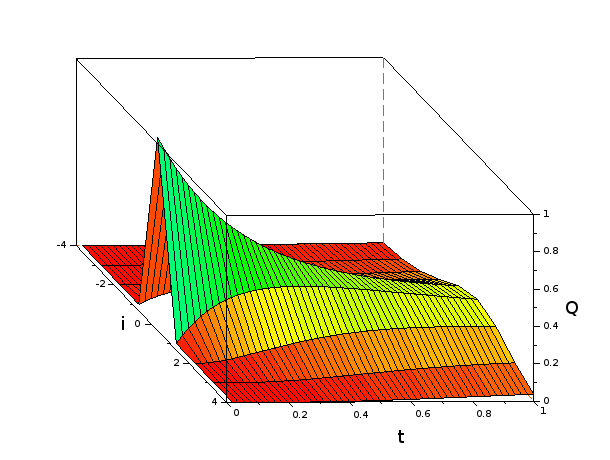}
							 \caption{Evolution of the  distribution of the CTOQW starting from a state localized in 0. 
							 					The $i$-axis stands for the position $|i\rangle$ on $\Z$, 
							 					the $t$-axis stands for the time 
							 					and the $Q$-axis returns the distribution $q_t(i)$.}
							\label{f1}
							\end{center}\end{figure}
							
							One can check that the condition of Definition \ref{irr} is satisfied which implies that $\phi$ is irreducible. Hence the process $\Big(\frac{X_t-X_0}{t}\Big)_{t\geq0} $ satisfies a LDP with a good rate function $\Lambda^*:x\mapsto\sup_{u\in\R}(u.x-l_u)$ (see Figure \ref{f2} for numerical computations).
							
		\begin{figure}
		
				\begin{center}
					\includegraphics[width = 120mm]{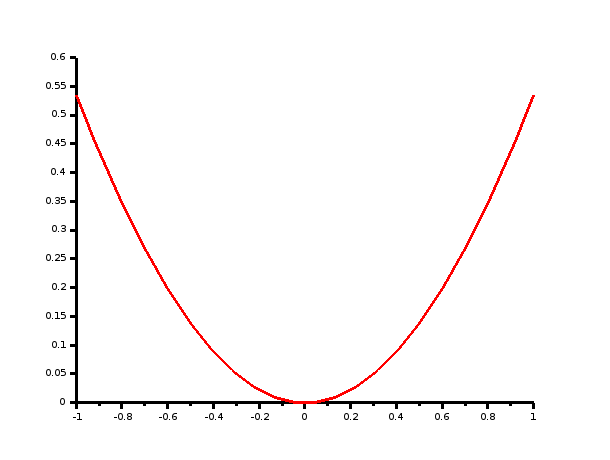}
					 \caption{$\Lambda^*$ for the first example.}
					 \label{f2}
				\end{center}
				\end{figure}

\item For the second example on $\Z$, we consider the following operators
	\[D_0=\begin{pmatrix}
								-\frac{3}{8}&0\\
								0&-\frac{1}{4}
								\end{pmatrix}, \ D_1=\begin{pmatrix}
					0&\frac{1}{2}\\
					\frac{1}{2}&0
					\end{pmatrix} \text{ and }D_{2}=\begin{pmatrix}
					0&\frac{1}{2}\\
					\frac{1}{\sqrt{2}}&0
					\end{pmatrix} \, .\] 
				Concerning the invariant state, after easy computations, we get
								
										\[
									\Lin(\rho)=0
								\iff \rho=\begin{pmatrix}
																\frac{2}{5}&0\\
																0&\frac{3}{5}
																							\end{pmatrix} \, .\]
This implies that $(H_1)$ is satisfied with
						\[\rho_{inv}=\begin{pmatrix}
								\frac{2}{5}&0\\
								0&\frac{3}{5}
															\end{pmatrix} \, .\]
													We can then compute the following quantities involved in the CLT:
													\[m=\Tr(D_1\rho_{{inv}} D_1^*)-\Tr(D_{2}\rho_{{inv}} D_{2}^*)=-\frac{1}{10} \ ; \ 
												J=\frac{1}{10}\begin{pmatrix}
													-1& 0\\
													0 & 1
														\end{pmatrix} \ ; \ \sigma^2=\frac{73}{125}\, .\]
												
			The CLT yields:
												\[\dfrac{X_t+\frac{t}{10}}{\sqrt{t}} \underset{t \to +\infty}{\overset{\Lin}{\longrightarrow}}\mathcal{N}\left(0,\frac{73}{125}\right) \, .\]
												
The reader can easily check that the condition of Definition \ref{irr} is satisfied, then $\phi$ is irreducible. Hence the process $\Big(\frac{X_t-X_0}{t}\Big)_{t\geq0}$ satisfies a LDP with a good rate function $\Lambda^*:x\mapsto\sup_{u\in\R}(u.x-l_u)$. In this case we are able to compute explicitly $l_u$. Recall that $l_u=\max\{\Re(\lambda),\lambda\in\Sp(\mathcal L^{(u)})\}$, where $\mathcal L^{(u)}$ is given by:
	\[\mathcal L^{(u)}=\frac{1}{8}\begin{pmatrix}
	-6& 0& 0& 2(e^u+e^{-u})\\
	0&-5&2(e^u+\sqrt{2}e^{-u})&0\\
	0&2(e^u+\sqrt{2}e^{-u})&-5&0\\
	2(e^u+2e^{-u})&0&0&-4
	
	\end{pmatrix} \, .\]
	
	Hence, tedious computations show that
	$l_u=\frac{1}{32}\Big(20+\sqrt{208+64 e^{2u}+128 e^{-2u}} \Big)$. Figure \ref{f3} displays the good rate function $
	\Lambda^*$.
	\begin{figure}[h!]
					\begin{center}
						\includegraphics[width = 120mm]{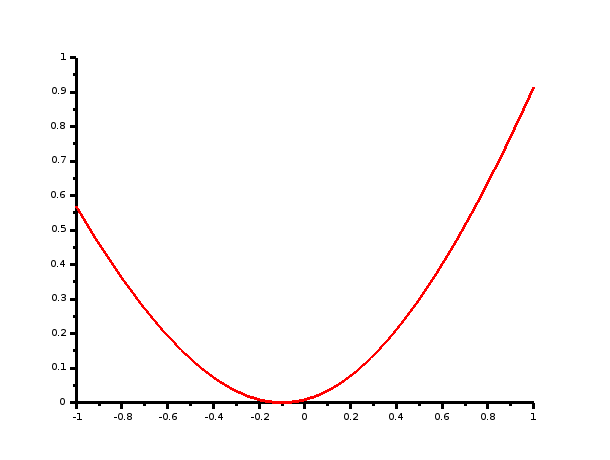}
						 \caption{$\Lambda^*$ for the second example.}
						 \label{f3}
					\end{center}
					\end{figure}

\item Now, we study an example of CTOQWs on $\Z^2$. We choose the following five operators:
\[D_0=\begin{pmatrix}
								-\frac{1}{2}&0\\
								0&-\frac{3}{8}
								\end{pmatrix}, \ D_1=\frac{1}{\sqrt{6}}\begin{pmatrix}
					1&1\\
					0&1
					\end{pmatrix}, \ D_{2}=\frac{1}{2\sqrt{2}}\begin{pmatrix}
																																		0&1\\
																																		0&1
																																		\end{pmatrix}, \]
																																		\[
																D_{3}=\frac{1}{\sqrt{6}}\begin{pmatrix}
																																							1&0\\
																																							-1&1
																																							\end{pmatrix}  
										 \text{ and } D_{4}=\frac{1}{\sqrt{2}}\begin{pmatrix}
										 										1&0\\
										 										0&0
										 										\end{pmatrix} \, ,\]		

which satisfie \[D_0+D_0^*+\sum_{r=1}^4 D_r^*D_r=0 \, .\]
This case is illustrated numerically by Figure \ref{f4}.

\begin{figure}[h]
	\begin{center}

																						\includegraphics[width=140mm,height=80mm]{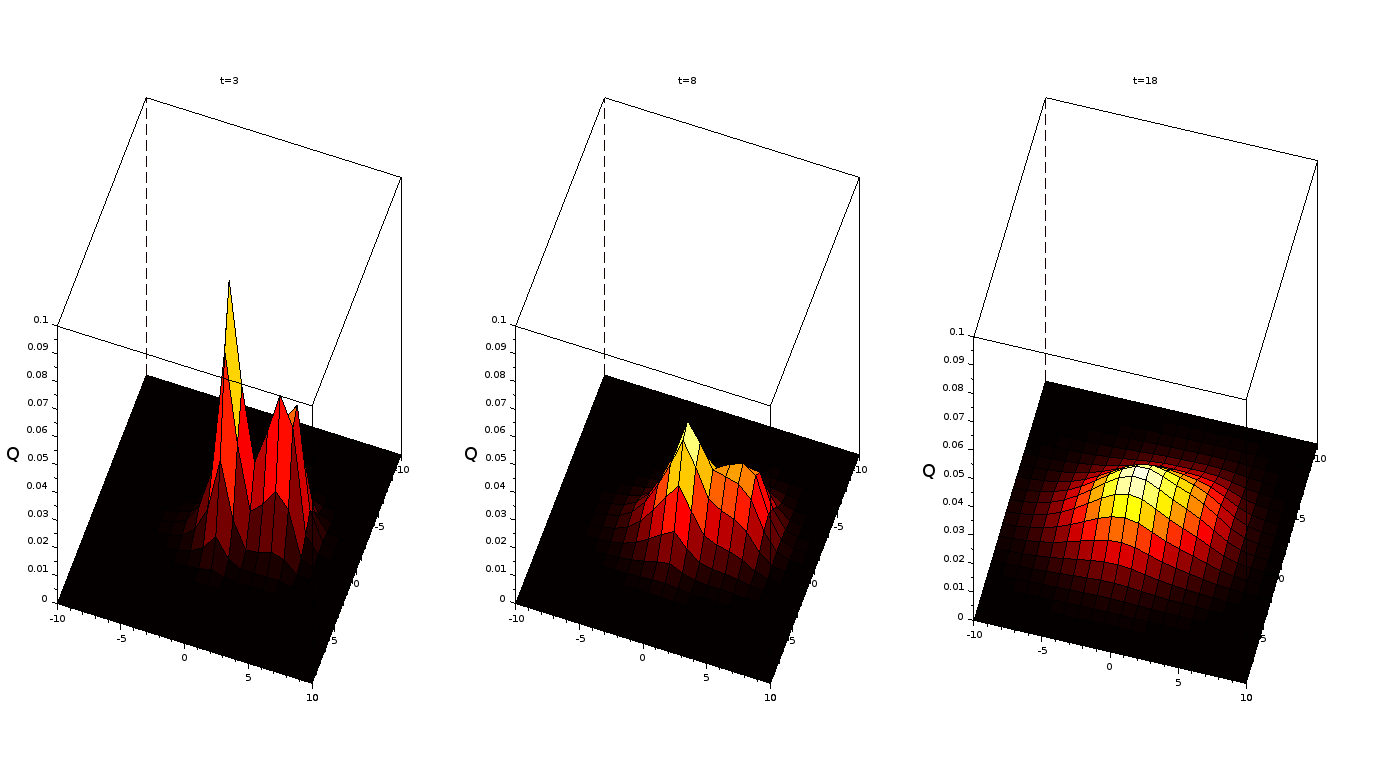}
				\caption{{Distribution of the CTOQW on $\Z^2$ at time t = 3, 8 and 18. }}
						\label{f4}						\end{center}	\end{figure}

Condition $(H1)$ is satisfied with $\rho_{inv}=\begin{pmatrix}
															 																						 										\frac{7}{11}&0\\
															 																						 										0&\frac{4}{11}
															 																						 										\end{pmatrix}$. The quantities for the CLT are
					\[\begin{split}
					m=\left(-\frac{1}{22},-\frac{5}{22}\right),& \ J_1=\frac{4}{33}\begin{pmatrix}	-5 & 2\\
					2 & 5\end{pmatrix},\\
					  J_2=\frac{3}{77}\begin{pmatrix}	-13 & -8\\
							-8 & 13\end{pmatrix}   \text{ and }&
					V=\frac{1}{23958} \begin{pmatrix}	10651 &   -414  \\-414 & 14661	\end{pmatrix} \, .
										\end{split}\]
					 The Central Limit Theorem states then
					 				\[\dfrac{X_t-m t}{\sqrt{t}} \underset{t \to +\infty}{\overset{\Lin}{\longrightarrow}}\mathcal{N}(0,V) \, . \]
 
 Again, the conditon of Definition \ref{irr} is satisfied, then $\phi$ is irreducible. Therefore the process $\Big(\frac{X_t-X_0}{t}\Big)_{t\geq0}$ satisfies a LDP with a good rate function $\Lambda^*:x\mapsto\sup_{u\in\R^d}(u.x-l_u)$. In this case, the matrix $\mathcal L^{(u)}$ is given by:
 	\[ \mathcal L^{(u)}=\frac{1}{24}\left(\begin{smallmatrix}
 	4(-6+e^{u_1}+e^{-u_1}+3e^{-u_2})& 4e^{u_1}& 4e^{u_1}& 4e^{u_1}+3e^{u_2}\\
 	-4e^{-u_1}&-21+4(e^{u_1}+e^{-u_1})  &0 & 4e^{u_1}+3e^{u_2}\\
 -4e^{-u_1}&0&-21+4(e^{u_1}+e^{-u_1})&4e^{u_1}+3e^{u_2}\\
 4e^{-u_1}&-4e^{-u_1}&-4e^{-u_1}&-18+4e^{-u_1}+3e^{u_2}	
 	\end{smallmatrix}\right) \, .\]
 	We are not able to obtain an analytic expression of $l_u$, we then plot an numerical approximation of $\Lambda^*$ (see Figure \ref{f5}).

\end{enumerate}

\begin{figure}[h!]	
					\begin{center}
						\includegraphics[width = 110mm]{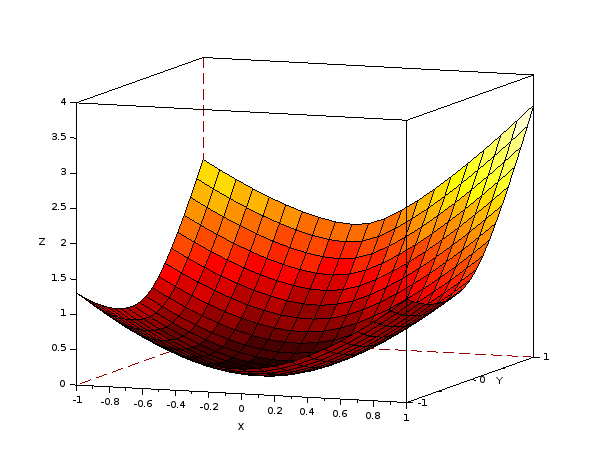}
						 \caption{{$\Lambda^*$ for the last example.}}\label{f5}
					\end{center}
					\end{figure}
			
				\textbf{Acknowledgments.} 
				The author wishes to thank his PhD advisor C. Pellegrini for suggesting
				the problem and his sustained help.
				The author is also grateful to Y. Pautrat for interesting discussions around the topics and examples of the present manuscript. This work is financially supported by the ANR project STOQ: ANR-14-CE25-0003. The author wants to thank the organizers of the summer school "stochastic methods in Quantum Mechanics" in Autrans (2016) where this works have been completed. I thank the referee for helpful comments in improving the exposition.
			\clearpage		
\nocite{*}
\bibliographystyle{plain}
\bibliography{bibtcloqwr}

\begin{thebibliography}{10}

\bibitem{MR3296637}
S.~Attal, N.~Guillotin-Plantard, and C.~Sabot.
\newblock Central limit theorems for open quantum random walks and quantum
  measurement records.
\newblock {\em Ann. Henri Poincar\'e}, 16(1):15--43, 2015.

\bibitem{MR2930584}
S.~Attal, F.~Petruccione, C.~Sabot, and I.~Sinayskiy.
\newblock Open quantum random walks.
\newblock {\em J. Stat. Phys.}, 147(4):832--852, 2012.

\bibitem{barchielli2009quantum}
A.~Barchielli and M.~Gregoratti.
\newblock {\em Quantum trajectories and measurements in continuous time: the
  diffusive case}, volume 782.
\newblock Springer, 2009.

\bibitem{MR2759470}
A.~Barchielli and C.~Pellegrini.
\newblock Jump-diffusion unravelling of a non-{M}arkovian generalized
  {L}indblad master equation.
\newblock {\em J. Math. Phys.}, 51(11):112104, 18, 2010.

\bibitem{MR3268105}
M.~Bauer, D.~Bernard, and A.~Tilloy.
\newblock The open quantum {B}rownian motions.
\newblock {\em J. Stat. Mech. Theory Exp.}, (9):p09001, 48, 2014.

\bibitem{breuer2007non}
H.-P. Breuer.
\newblock Non-{M}arkovian generalization of the {L}indblad theory of open
  quantum systems.
\newblock {\em Phys. Rev. A (3)}, 75(2):022103, 9, 2007.

\bibitem{carbone2014homogeneous}
R.~Carbone and Y.~Pautrat.
\newblock Homogeneous open quantum random walks on a lattice.
\newblock {\em J. Stat. Phys.}, 160(5):1125--1153, 2015.

\bibitem{carbone2014open}
R.~Carbone and Y.~Pautrat.
\newblock Open quantum random walks: reducibility, period, ergodic properties.
\newblock {\em Ann. Henri Poincar\'e}, 17(1):99--135, 2016.

\bibitem{MR551466}
E.~Christensen and D.~Evans.
\newblock Cohomology of operator algebras and quantum dynamical semigroups.
\newblock {\em J. London Math. Soc. (2)}, 20(2):358--368, 1979.

\bibitem{crimaldi2005convergence}
I.~Crimaldi and L.~Pratelli.
\newblock Convergence results for multivariate martingales.
\newblock {\em Stochastic Process. Appl.}, 115(4):571--577, 2005.

\bibitem{MR0269241}
E.~B. Davies.
\newblock Quantum stochastic processes. {II}.
\newblock {\em Comm. Math. Phys.}, 19:83--105, 1970.

\bibitem{dembo2009large}
A.~Dembo and O.~Zeitouni.
\newblock {\em Large deviations techniques and applications}, volume~38 of {\em
  Stochastic Modelling and Applied Probability}.
\newblock Springer-Verlag, Berlin, 2010.

\bibitem{ethier1986markov}
S.~Ethier and T.~Kurtz.
\newblock {\em Markov processes}.
\newblock Wiley Series in Probability and Mathematical Statistics: Probability
  and Mathematical Statistics. John Wiley \& Sons, Inc., New York, 1986.

\bibitem{evans1978spectral}
D.~Evans and R.~H{\o}egh-Krohn.
\newblock Spectral properties of positive maps on {$C\sp*$}-algebras.
\newblock {\em J. London Math. Soc. (2)}, 17(2):345--355, 1978.

\bibitem{MR2432029}
S.~Gudder.
\newblock Quantum {M}arkov chains.
\newblock {\em J. Math. Phys.}, 49(7):072105, 14, 2008.

\bibitem{MR624435}
P.~Hall and C.~C. Heyde.
\newblock {\em Martingale limit theory and its application}.
\newblock Academic Press, Inc. [Harcourt Brace Jovanovich, Publishers], New
  York-London, 1980.

\bibitem{jakvsic2014entropic}
V.~Jak{\v{s}}i{\'c}, C.-A. Pillet, and M.~Westrich.
\newblock Entropic fluctuations of quantum dynamical semigroups.
\newblock {\em J. Stat. Phys.}, 154(1-2):153--187, 2014.

\bibitem{kato2013perturbation}
T.~Kato.
\newblock {\em Perturbation theory for linear operators}.
\newblock Classics in Mathematics. Springer-Verlag, Berlin, 1995.

\bibitem{MR2106330}
B.~K{\"u}mmerer and H.~Maassen.
\newblock A pathwise ergodic theorem for quantum trajectories.
\newblock {\em J. Phys. A}, 37(49):11889--11896, 2004.

\bibitem{MR1887206}
A.~Makowski and A.~Shwartz.
\newblock The {P}oisson equation for countable {M}arkov chains: probabilistic
  methods and interpretations.
\newblock In {\em Handbook of {M}arkov decision processes}, volume~40 of {\em
  Internat. Ser. Oper. Res. Management Sci.}, pages 269--303. Kluwer Acad.
  Publ., Boston, MA, 2002.

\bibitem{marais2013decoherence}
A.~Marais, I.~Sinayskiy, A.~Kay, F.~Petruccione, and A.~Ekert.
\newblock Decoherence-assisted transport in quantum networks.
\newblock {\em New J. Phys.}, 15(January):013038, 18, 2013.

\bibitem{motwani2010randomized}
R.~Motwani and P.~Raghavan.
\newblock {\em Randomized algorithms}.
\newblock Cambridge University Press, Cambridge, 1995.

\bibitem{MR2478685}
C.~Pellegrini.
\newblock Existence, uniqueness and approximation of a stochastic
  {S}chr\"odinger equation: the diffusive case.
\newblock {\em Ann. Probab.}, 36(6):2332--2353, 2008.

\bibitem{pellegrini2014continuous}
C.~Pellegrini.
\newblock Continuous time open quantum random walks and non-{M}arkovian
  {L}indblad master equations.
\newblock {\em J. Stat. Phys.}, 154(3):838--865, 2014.

\bibitem{MR2546050}
C.~Pellegrini and F.~Petruccione.
\newblock Non-{M}arkovian quantum repeated interactions and measurements.
\newblock {\em J. Phys. A}, 42(42):425304, 21, 2009.

\bibitem{MR1725357}
D.~Revuz and M.~Yor.
\newblock {\em Continuous martingales and {B}rownian motion}, volume 293 of
  {\em Grundlehren der Mathematischen Wissenschaften}.
\newblock Springer-Verlag, Berlin, third edition, 1999.

\bibitem{sinayskiy2012efficiency}
I.~Sinayskiy and F.~Petruccione.
\newblock Efficiency of open quantum walk implementation of dissipative quantum
  computing algorithms.
\newblock {\em Quantum Inf. Process.}, 11(5):1301--1309, 2012.

\end{thebibliography}

\end{document}